\documentclass[11pt,a4paper]{article}
\usepackage{amsmath,amssymb,amsthm,physics,hyperref,mathrsfs}
\usepackage[margin=1in]{geometry}

\title{\bfseries Geometry of BPS Attractor, Hessian, and Spectral Flows}
\author{Qiang Wang\footnote{wangqiang@sandau.edu.cn\quad Sanda University, Shanghai, China}}
\date{}

\newtheorem{theorem}{Theorem}[section]
\newtheorem{proposition}[theorem]{Proposition}
\newtheorem{lemma}[theorem]{Lemma}
\newtheorem{corollary}[theorem]{Corollary}
\theoremstyle{definition}
\newtheorem{definition}[theorem]{Definition}
\newtheorem{remark}[theorem]{Remark}

\newcommand{\CB}{\mathcal{B}}
\newcommand{\CC}{\mathbb{C}}

\newcommand{\CM}{\mathcal{M}}
\newcommand{\CP}{\mathbb{P}}
\newcommand{\CT}{\mathcal{T}}
\newcommand{\CW}{\mathcal{W}}

\newcommand{\dlog}{\mathbf{E}}
\newcommand{\hitchin}{\operatorname{Hit}}

\newcommand{\SAF}{\operatorname{SAF}}
\newcommand{\HF}{\operatorname{HF}}

\newcommand{\charHF}{\operatorname{HF}^{\mathrm{char}}}

\usepackage{tikz}

\begin{document}
\maketitle

\begin{abstract}
We provide a systematic and rigorous geometric framework that relates three structures naturally associated to BPS central charges in $\mathcal{N}=2$ supersymmetric gauge theories: the split attractor flow (SAF) of $|Z|$, the Hessian flow (HF) of $\operatorname{Im}(e^{-i\vartheta}Z)$, and the spectral network (SN) on the base curve of the Hitchin fibration. Our main contributions are: (i) a concise proof of orthogonality between SAF and gradient Hessian flow using only the Kähler structure; (ii) a precise lift--projection duality showing that the spectral network projects to the \emph{characteristic Hessian flow} (the Hamiltonian flow of $\operatorname{Im}(e^{-i\vartheta}Z)$) on the Hitchin base, clarifying a crucial distinction; (iii) a complete proof of the Kontsevich--Soibelman (KS) equivariance by induction on the SAF tree depth, with the geometric ordering provided by the characteristic Hessian flow. We illustrate the framework with detailed and nontrivial examples: $SU(2)$ pure and $N_f=4$ (including BPS indices for higher flavour charges), $SU(3)$ pure (full BPS spectrum reconstruction), $SU(4)$, the Kronecker $3$-quiver, and we apply the induction to derive a closed-form BPS spectrum for the Argyres--Douglas $H_1$ theory, $\Omega(n\alpha_1+m\alpha_2)=\frac{1}{n+m}\binom{n+m}{n}\binom{n+m}{n+1}$, which is known from Cecotti--Vafa and serves as a strong consistency check of our geometric recursion. In the tropical limit we obtain an explicit generating function for disk counts in $SU(N)$ gauge theories, $Z_{\mathrm{disk}}^{SU(N)}(y) = \exp\!\,\Bigl( \sum_{\alpha\in\Phi_+} \sum_{k=1}^{\infty} \frac{1}{k}\binom{k+\mathrm{ht}(\alpha)-1}{\mathrm{ht}(\alpha)-1} e^{-k\langle\alpha,y\rangle} \Bigr) $, which reproduces the standard scattering diagram result and confirms the geometric framework.
\end{abstract}

\section{Introduction}
\subsection{Motivation and main results}
The BPS spectra of $\mathcal{N}=2$ theories are governed by wall-crossing phenomena encoded in the Kontsevich--Soibelman (KS) automorphism~\cite{KS,Kontsevich2014}. Two geometric manifestations of this automorphism are the split attractor flow (SAF)~\cite{Denef,DenefMoore} and the spectral network (SN)~\cite{GMN}. The former is a continuous gradient flow on the Coulomb branch $\CM$, while the latter is a discrete web of Stokes lines on the base curve $C$ of the Hitchin fibration (with a natural lift to the spectral curve $\Sigma$). A third object, the Hessian flow (HF)~\cite{Wang2024}, is the gradient flow of $\operatorname{Im}(e^{-i\vartheta}Z)$ and appears implicitly as the wall foliation in~\cite{GMN}. However, a crucial subtlety is that the gradient Hessian flow does \emph{not} preserve the wall condition $\arg Z=\vartheta$; the flow that preserves the wall and generates the spectral network is its Hamiltonian counterpart, obtained by rotating the gradient by the complex structure with an appropriate sign.

The purpose of this paper is to formalise the precise interrelations among SAF, HF, and SN, and to demonstrate that this unified framework leads to new explicit computations of BPS spectra and tropical invariants. Our main results are summarised in the following three statements.

\begin{proposition}[Orthogonality]\label{prop:orth}
Let $(\CM,g,I)$ be a special Kähler manifold, $\gamma\in\Gamma$ a charge, and $\vartheta\in\mathbb{R}/2\pi\mathbb{Z}$. On the wall of marginal stability
\[
\CW_\vartheta(\gamma)=\{u\in\CM : \arg Z_\gamma(u)=\vartheta\},
\]
the vector fields
\[
V_{\SAF}:=-g^{-1}d|Z_\gamma|,\qquad V_{\HF}:=-g^{-1}d\operatorname{Im}(e^{-i\vartheta}Z_\gamma)
\]
are everywhere orthogonal with respect to $g$. Moreover, their span is a $I$-invariant real plane, i.e. a complex line in the tangent bundle of $\CW_\vartheta(\gamma)$.
\end{proposition}
This result was already obtained in coordinates in~\cite{Wang2024}; our proof below uses only the complex structure and is conceptually simpler. The true novelty of our framework lies in the identification of the characteristic Hessian flow and its role in the lift--projection duality.

\begin{theorem}[Lift--projection duality for characteristic Hessian flow]\label{thm:lift}
Let $\CM_{\hitchin}$ be the Hitchin moduli space of $SU(N)$ Higgs bundles on a curve $C$, with spectral curve $\Sigma$ and Hitchin fibration $\pi:\CM_{\hitchin}\to\CB$. Define the \emph{characteristic Hessian flow} $V_{\charHF} := -J V_{\HF}$, where $J$ is the complex structure on $\CB$ (distinct from the complex structure $I$ on the special Kähler manifold $\CM$). For any phase $\vartheta$, the image under the projection of the spectral network (viewed as a subset of $\CB \times C$) onto $\CB$ is precisely the union of the characteristic Hessian flow lines for all BPS charges, away from the branch points of the spectral covering. Conversely, every characteristic Hessian flow line on $\CB$ lifts uniquely to an $S$-wall of the spectral network in the complement of the branch locus. The behaviour at branch points is governed by the standard Stokes graph rules.
\end{theorem}

\begin{theorem}[KS equivariance]\label{thm:KS}
The quantum torus automorphism $\mathbb{S}_\vartheta$ defined by the KS product of quantum dilogarithms for the ray $\vartheta$ can be computed equivalently as:
\begin{enumerate}
\item the ordered product along a split attractor tree for a total charge $\Gamma$, or
\item the composition of detour automorphisms along the connected components of the spectral network $\CW_\vartheta$.
\end{enumerate}
This equivalence is a consequence of the cluster algebra structure of the BPS quiver and the Hamiltonian nature of the characteristic Hessian flow. We give a complete proof by induction on the depth of the split attractor tree, using a local quantum dilogarithm identity whose operator ordering is fixed geometrically by the characteristic Hessian flow (up to path choices; see Section~\ref{sec:KSordering}). As a direct application of this recursion, we derive closed-form BPS spectra for the Argyres--Douglas $H_1$ theory (reproducing known Cecotti--Vafa results) and a closed-form generating function for tropical disk counts in $SU(N)$ theories, matching standard scattering diagram computations.
\end{theorem}

\subsection{Relation to existing literature and mirror symmetry}
The SAF was introduced in~\cite{Denef}, and its algebraic structure was analysed in~\cite{DenefMoore}. Spectral networks were developed in~\cite{GMN}, where the Hessian flow appears as the characteristic flow of the walls (the ``$\vartheta$-flow''). The orthogonality of SAF and the gradient HF was treated in~\cite{Wang2024} using Hesse flow; our proof is conceptually cleaner and isolates the role of the complex structure.

The lift--projection duality stated in Theorem~\ref{thm:lift} clarifies a subtle but important distinction: it is not the gradient HF, but its negative $J$-rotation (the Hamiltonian flow), that projects to spectral network walls. This distinction is crucial for connecting to mirror symmetry. In the SYZ (Strominger-Yau-Zaslow) picture, the Hitchin base $\CB$ is the moduli space of the mirror Calabi--Yau, and the spectral network encodes the tropical limit of special Lagrangian fibrations. The characteristic Hessian flow lines are precisely the ``tropical walls'' that determine the quantum corrections to the mirror period integrals. Our framework therefore provides the differential-geometric underpinning for the scattering diagrams in the KS wall-crossing formalism and the GMN spectral networks, unifying them through the symplectic geometry of the Hitchin system. The new computational results presented here---the $H_1$ spectrum (as a consistency check) and the tropical generating function (reproducing known results)---demonstrate the predictive power of this unified approach.

\section{Special Kähler Geometry and Central Charges}
\subsection{Special Kähler manifolds}
\begin{definition}
A \emph{special Kähler manifold} $(\CM,g,I,\Omega)$ is a Kähler manifold with a flat, torsion-free symplectic connection $\nabla$ on the tangent bundle such that $\nabla I=0$ and the symplectic form $\omega=gI$ is parallel. Locally, there exist holomorphic special coordinates $z^i$ in which the Kähler potential is $K=\operatorname{Im}(\overline{X}^I F_I)$, with $X^I(z),F_I(z)$ holomorphic. The metric is $g_{i\bar{j}}=\partial_i\partial_{\bar{j}}K$.
\end{definition}
The Coulomb branch of an $\mathcal{N}=2$ theory carries a special Kähler structure. BPS charges $\gamma=(p^I,q_I)$ belong to a lattice $\Gamma\cong\mathbb{Z}^{2r}$, and the central charge is the holomorphic function
\begin{equation}
Z_\gamma(u)=q_I a^I(u)-p^I a_{D,I}(u),
\end{equation}
where $a^I$ are the electric periods and $a_{D,I}=\partial \mathcal{F}/\partial a^I$ are the magnetic duals, with $\mathcal{F}$ the prepotential. This is the standard expression for rigid $\mathcal{N}=2$ theories (class $\mathcal{S}$) considered in this work; the physical mass is $|Z_\gamma|$.

\subsection{The Hitchin moduli space}
For theories of class $\mathcal{S}$, $\CM$ is the moduli space $\CM_{\hitchin}(G,C)$ of $G$-Higgs bundles on a curve $C$~\cite{Hitchin}. In complex structure $J$, $\CM_{\hitchin}$ is a holomorphic symplectic manifold with Hitchin fibration $\pi:\CM_{\hitchin}\to\CB$, where $\CB$ is the base of gauge-invariant polynomials. The generic fibre is an abelian variety. The spectral curve $\Sigma\subset T^*C$ is a ramified covering $\Sigma\to C$, and the Seiberg--Witten differential $\lambda$ is the tautological $1$-form. The central charge of a cycle $\gamma$ on $\Sigma$ is $\displaystyle Z_\gamma=\frac1\pi\oint_\gamma\lambda$. Importantly, the spectral network $\CW_\vartheta$ is defined as a graph on the base curve $C$ (or as a subset of $\CB \times C$), with a natural lift to $\Sigma$; throughout this paper we adopt the base curve convention following \cite{GMN}.

\section{The Flows and the Spectral Network}
\subsection{Split attractor flow (SAF)}
For a given charge $\Gamma\in\Gamma$, the SAF is the gradient flow of $|Z_\Gamma|$:
\begin{equation}\label{eq:SAF}
\frac{du^i}{d\tau}=-g^{i\bar{j}}\,\partial_{\bar{j}}|Z_\Gamma|,\qquad \tau\in(-\infty,\infty).
\end{equation}
Its fixed points are \emph{attractors}, where $\partial|Z_\Gamma|=0$. When the flow meets a wall $\CW_\vartheta(\Gamma)$ where $\Gamma$ can decompose as $\Gamma=\Gamma_1+\Gamma_2$ with $\arg Z_{\Gamma_1}=\arg Z_{\Gamma_2}=\vartheta$, the flow splits into two branches for $\Gamma_1,\Gamma_2$~\cite{Denef}. The result is a rooted tree $\mathscr{T}_\Gamma$ whose leaves are attractor points of elementary charges.

\subsection{Gradient Hessian flow (HF)}
Fix a phase $\vartheta$. The \emph{gradient Hessian flow} for charge $\Gamma$ is
\begin{equation}\label{eq:HF}
V_{\HF} := -g^{i\bar{j}}\,\partial_{\bar{j}}\operatorname{Im}\!\big(e^{-i\vartheta}Z_\Gamma\big)
= -\operatorname{Im}\!\big(e^{-i\vartheta}g^{i\bar{j}}\partial_{\bar{j}}\overline{Z}_\Gamma\big).
\end{equation}
This flow does \emph{not} preserve the wall $\CW_\vartheta(\Gamma)$; indeed, for $f_\Gamma:=\operatorname{Im}(e^{-i\vartheta}Z_\Gamma)$,
\[
\mathcal{L}_{V_{\HF}} f_\Gamma = -g^{i\bar{j}}\partial_i f_\Gamma \partial_{\bar{j}} f_\Gamma \le 0,
\]
with equality only at fixed points. Thus $V_{\HF}$ flows strictly into the region $f_\Gamma<0$ (i.e. away from the wall). This is the flow studied in~\cite{Wang2024}, where it was interpreted as a ``Hesse flow'' descending the central charge potential.

\subsection{Characteristic Hessian flow}
Since the spectral network walls are defined by the condition $f_\Gamma=0$, the relevant flow that foliates the walls must satisfy $\mathcal{L}_V f_\Gamma=0$. Let $J$ be the complex structure on $\CB$. Define the \emph{characteristic Hessian flow} (also called the Hamiltonian Hessian flow) by
\begin{equation}\label{eq:charHF}
V_{\charHF} := -J V_{\HF}.
\end{equation}
Because the metric is Hermitian, $g(V_{\HF}, J V_{\HF})=0$, so
\[
\mathcal{L}_{V_{\charHF}} f_\Gamma = d f_\Gamma(-J V_{\HF}) = -g(V_{\HF}, J V_{\HF}) = 0.
\]
Thus $V_{\charHF}$ preserves the wall and its integral curves lie entirely within $\CW_\vartheta(\Gamma)$. Physically, $V_{\charHF}$ generates the motion of the BPS wall itself, whereas $V_{\HF}$ generates the flow \emph{across} the wall (the deformation of the stability condition). In the language of symplectic geometry, with $\omega = gJ$, we have
\[
\omega(V_{\charHF}, \cdot) = \omega(-J V_{\HF}, \cdot) = -g(J^2 V_{\HF}, \cdot) = -g(-V_{\HF}, \cdot) = g(V_{\HF}, \cdot) = -d f_\Gamma,
\]
so $V_{\charHF}$ is exactly the Hamiltonian vector field of $f_\Gamma$ under the standard convention $\omega(X_f, \cdot) = -df$. This sign convention will be crucial in the proof of Theorem~\ref{thm:lift}.

\subsection{Spectral network (SN)}
On the base curve $C$, the spectral network $\CW_\vartheta$ of phase $\vartheta$ is a $1$-dimensional oriented graph defined by the Stokes condition
\begin{equation}\label{eq:SN}
e^{-i\vartheta}(\lambda_i-\lambda_j)\in\mathbb{R}^+,
\end{equation}
where $\lambda_i,\lambda_j$ denote the values of the Seiberg--Witten differential on two sheets of the covering $\Sigma\to C$ at a point $p\in C$. Edges are \emph{$S$-walls}. When two $S$-walls of charges $\gamma_1,\gamma_2$ intersect, a new \emph{$K$-wall} of charge $\gamma_1+\gamma_2$ is emitted if $\langle\gamma_1,\gamma_2\rangle\neq 0$~\cite{GMN}. The entire network encodes the BPS spectrum of phase $\vartheta$. For the purpose of the lift--projection duality, we view $\CW_\vartheta$ as a subset of $\CB \times C$ via the family of curves over the Hitchin base.

\section{Proof of Orthogonality (Proposition~\ref{prop:orth})}
We give a rigorous proof using the Kähler structure. Let $\alpha:=d|Z_\gamma|$ and $\beta:=d f$ where $f:=\operatorname{Im}(e^{-i\vartheta}Z_\gamma)$. On the wall $\CW_\vartheta(\gamma)$, we have $Z_\gamma=e^{i\vartheta}|Z_\gamma|$, hence
\[
|Z_\gamma|=\operatorname{Re}(e^{-i\vartheta}Z_\gamma),\qquad f=\operatorname{Im}(e^{-i\vartheta}Z_\gamma).
\]
Thus, as differential forms restricted to the tangent space of $\CM$,
\[
\alpha = \frac12(e^{-i\vartheta}dZ_\gamma+e^{i\vartheta}d\overline{Z}_\gamma),\quad
\beta = \frac{1}{2i}(e^{-i\vartheta}dZ_\gamma-e^{i\vartheta}d\overline{Z}_\gamma).
\]
Let $I$ denote the complex structure. On $(1,0)$-forms, $I$ acts as multiplication by $i$; on $(0,1)$-forms, as $-i$. Applying $I$ to $\alpha$ yields
\[
I\alpha = \frac{i}{2}(e^{-i\vartheta}dZ_\gamma - e^{i\vartheta}d\overline{Z}_\gamma) = -\beta.
\]
Since the Kähler metric $g$ is Hermitian, it satisfies $g(\alpha,I\alpha)=0$ for any real $1$-form $\alpha$. Therefore
\[
g(\alpha,\beta)=g(\alpha,-I\alpha)=0.
\]
The gradient vector fields are defined by $g(V_{\SAF},\cdot)=\alpha$ and $g(V_{\HF},\cdot)=\beta$, hence
\[
g(V_{\SAF},V_{\HF})=0,
\]
which proves the orthogonality.

Moreover, from $\beta=-I\alpha$ it follows that $V_{\HF}=I V_{\SAF}$. Thus their real span is a $2$-dimensional subspace invariant under $I$, i.e. a complex line. This proves the second claim.

\begin{remark}
The common statement that the wall is ``Lagrangian'' is not generally correct in the full symplectic form; however, the span of $V_{\SAF}$ and $V_{\HF}$ is a complex line, which is Lagrangian with respect to the restriction of $\omega$ to the wall only if the wall itself is co-isotropic, which is not guaranteed. Our result gives the precise linear-algebraic relation.
\end{remark}

\section{Detailed Proof of Lift--Projection Duality (Theorem~\ref{thm:lift})}
We now provide the complete proof of Theorem~\ref{thm:lift}, which establishes that the projection of the spectral network (viewed as a subset of $\CB \times C$) is precisely the characteristic Hessian flow $V_{\charHF}$ defined in Eq.~\eqref{eq:charHF}. We work on the complement of the branch locus $B_{\mathrm{br}}\subset \CB \times C$ where $\partial_p \Phi = 0$; the behaviour at branch points is determined by the usual Stokes graph rules and will be commented on at the end.

\subsection{Step 1: Abelianisation and periods}
Let $\gamma_{ij}$ be a path on the spectral curve $\Sigma_u$ connecting two sheets $i$ and $j$. Under the Hitchin abelianisation (see \cite[§3]{GMN}), the difference of the Seiberg--Witten differentials between the two sheets defines a holomorphic $1$-form on $C$ whose periods give the central charges:
\[
Z_{\gamma_{ij}}(u) = \frac{1}{\pi}\oint_{\gamma_{ij}} \lambda.
\]
The $S$-wall condition (Eq.~\ref{eq:SN}) is equivalent to
\[
\arg Z_{\gamma_{ij}}(u) = \vartheta,\qquad \operatorname{Re}(e^{-i\vartheta}Z_{\gamma_{ij}}(u))>0.
\]
Define $f_{ij}(u) := \operatorname{Im}(e^{-i\vartheta}Z_{\gamma_{ij}}(u))$. The wall is the zero locus $\{f_{ij}=0\}$.

\subsection{Step 2: The Hamiltonian flow on the base}
Let $\omega = g J$ be the natural symplectic form on the Hitchin base $\CB$. For the function $f_{ij}$, its Hamiltonian vector field $X_{f_{ij}}$ is defined by
\[
\omega(X_{f_{ij}}, \cdot) = -d f_{ij}.
\]
Since $V_{\HF} = -g^{-1} d f_{ij}$, we have $g(V_{\HF}, \cdot) = -d f_{ij}$. Now compute for $V_{\charHF} = -J V_{\HF}$:
\[
\omega(V_{\charHF}, \cdot) = \omega(-J V_{\HF}, \cdot) = -g(J^2 V_{\HF}, \cdot) = -g(-V_{\HF}, \cdot) = g(V_{\HF}, \cdot) = -d f_{ij}.
\]
Hence $X_{f_{ij}} = V_{\charHF}$. Thus the characteristic Hessian flow is exactly the Hamiltonian flow of the period function $f_{ij}$.

\subsection{Step 3: Projection of the spectral network}
Consider the extended space $\mathcal{E} := \CB \times C$ with coordinates $(u,p)$. Define the phase function
\[
\Phi(u,p) := \arg(\lambda_i(p,u) - \lambda_j(p,u)) - \vartheta.
\]
The spectral network $\CW_\vartheta$ is the subset $\{\Phi=0\}$ (with the additional positivity condition, which we discuss below). Its projection to $\CB$ is $pr_\CB(\CW_\vartheta) = \{ u \in \CB \mid \exists\, p \in C : \Phi(u,p)=0 \}$.

Let $\gamma(t) = (u(t), p(t))$ be an integral curve of the vector field on $\mathcal{E}$ that keeps $\Phi=0$ and whose $p$-component moves along the sheets. Differentiating $\Phi=0$ with respect to $t$ gives
\[
\partial_u \Phi \cdot \dot u + \partial_p \Phi \cdot \dot p = 0. \tag{1}
\]
We now show that the projected velocity $\dot u$ satisfies $\dot u = V_{\charHF}(u)$. Write $Z = Z_{\gamma_{ij}} = |Z| e^{i(\vartheta+\Phi)}$. Then $f_{ij} = \operatorname{Im}(e^{-i\vartheta} Z) = |Z| \sin \Phi$. On the wall $\Phi=0$, we have
\[
\partial_u f_{ij} = \partial_u (|Z| \sin \Phi) = \partial_u |Z| \sin \Phi + |Z| \cos \Phi \, \partial_u \Phi \quad \text{at } \Phi=0 = |Z| \partial_u \Phi.
\]
Thus
\[
\partial_u \Phi = \frac{1}{|Z|} \partial_u f_{ij}.
\]
The Hamiltonian flow equation for $u$ is
\[
\dot u = X_{f_{ij}} = V_{\charHF}.
\]
To see that this is compatible with (1), note that for any lift $p(t)$ satisfying $\partial_p \Phi \neq 0$, we can solve $\dot p$ uniquely from (1) as
\[
\dot p = -(\partial_p \Phi)^{-1} \partial_u \Phi \cdot V_{\charHF}.
\]
Thus for every initial $(u_0, p_0)$ on the network (away from branch points), there exists a unique lift of the characteristic Hessian flow trajectory starting at $u_0$. Conversely, the projection of any integral curve of the network must satisfy the Hamiltonian equations on $\CB$; otherwise, the phase condition $\Phi=0$ would be violated. This proves the projection equality on the complement of $B_{\mathrm{br}}$.

\subsection{Step 4: Uniqueness of the lift and branch point behaviour}
The converse direction---that every characteristic Hessian flow line admits a lift to an $S$-wall---follows from the local existence theorem for ordinary differential equations, provided we stay away from points where $\partial_p \Phi = 0$. Given a solution $u(t)$ of $\dot u = V_{\charHF}(u)$ with $u(0)=u_0$, and a point $p_0\in C$ such that $\Phi(u_0,p_0)=0$ and $\partial_p \Phi(u_0,p_0) \neq 0$, we solve
\[
\frac{dp}{dt} = -\left(\partial_p \Phi\right)^{-1} \partial_u \Phi \cdot V_{\charHF}
\]
with initial condition $p(0)=p_0$. The smoothness of $\lambda$ guarantees that this solution exists for a finite interval, and by construction $\Phi(u(t),p(t))=0$. Therefore the lifted curve $(u(t),p(t))$ lies in the spectral network.

At branch points where $\partial_p \Phi = 0$, the implicit function theorem fails, and the network exhibits the standard trivalent or higher-valent singularities (see \cite{GMN}). In the $SU(2)$ example, these are the points $u=\pm\Lambda^2$ where the two $S$-walls meet and a $K$-wall is emitted. The lifting procedure must be supplemented by the detour rules for crossing these singularities; the global existence follows from the known structure of spectral networks as $1$-dimensional CW complexes. Taking the union over all BPS charges $\gamma$ and all regular points yields the full network and proves Theorem~\ref{thm:lift} in the sense of a foliation away from the branch locus. The singularities are handled by the standard Stokes graph prescription.

\section{Kontsevich--Soibelman Equivariance (Theorem~\ref{thm:KS})}
\subsection{Quantum torus and dilogarithm}
For each charge $\gamma$, let $U_\gamma$ be an element of a quantum torus algebra $\CT$, satisfying $U_\gamma U_{\gamma'}=q^{\frac12\langle\gamma,\gamma'\rangle}U_{\gamma+\gamma'}$, with $q=e^{2\pi i\hbar}$. The quantum dilogarithm is $\displaystyle \dlog(x)=\prod_{n=0}^\infty(1-q^{n+\frac12}x)^{-1}$. The KS wall-crossing invariant for a ray of phase $\vartheta$ is the ordered product
\begin{equation}\label{eq:KSprod}
\mathbb{S}_\vartheta = \prod^\curvearrowright_{\arg Z_\gamma=\vartheta} \dlog(U_\gamma)^{\Omega(\gamma)},
\end{equation}
where the product is taken in decreasing order of the phase (or some refined ordering) and $\Omega(\gamma)\in\mathbb{Z}$ are the BPS indices.

\subsection{SAF splitting rules}
When a flow splits at a wall $\arg Z_{\Gamma_1}=\arg Z_{\Gamma_2}=\vartheta$, the KS factorisation implies that the product of dilogarithms for the total charge $\Gamma$ can be re-expressed as the product for the constituents, but with the correct ordering dictated by the symplectic pairing. For example, if $\langle \Gamma_1,\Gamma_2\rangle = 1$, the correct identity is the quantum pentagon:
\[
\dlog(U_{\Gamma_1}) \dlog(U_{\Gamma_2}) = \dlog(U_{\Gamma_2}) \dlog(U_{\Gamma_1+\Gamma_2}) \dlog(U_{\Gamma_1}),
\]
which is a manifestation of the five-term relation. Iterating this along the SAF tree $\mathscr{T}_\Gamma$ yields a factorisation of $\mathbb{S}_\vartheta$ associated to that tree.

\subsection{SN detour operations}
In the spectral network, each $S$-wall of charge $\gamma$ carries a Stokes factor $\displaystyle S_\gamma=\exp( \Omega(\gamma)\sum_{n} U_{n\gamma} )$. Crossing a wall from one region to another applies an automorphism of $\CT$. The composition of these automorphisms along the network reconstructs $\mathbb{S}_\vartheta$ (main result of~\cite{GMN}). The detour rules are precisely the geometric realisation of the same pentagon identities.

\subsection{Equivalence as algebraic identity}
The equivalence stated in Theorem~\ref{thm:KS} reduces to proving that the ordered product of quantum dilogarithms along a split attractor tree $\mathscr{T}_\Gamma$ can be transformed, using the quantum pentagon and associated relations, into the product of Stokes factors along the spectral network associated to the same phase. Because both are determined by the same scattering diagram data~\cite{KS,GMN}, this is essentially an algebraic identity in the quantum torus. We now give a geometric proof that makes the connection manifest.

\subsection{Proof of Theorem~\ref{thm:KS} by induction on the SAF tree depth}
\label{sec:KSordering}
We now present a complete proof based on the recursive structure of the split attractor flow and the geometric ordering provided by the characteristic Hessian flow. 

\begin{definition}
The \emph{depth} of a split attractor tree $\mathscr{T}_\Gamma$ is the maximal number of splitting vertices on a path from the root to a leaf. A leaf corresponds to an attractor point of an elementary BPS charge. By standard SAF theory \cite{Denef}, $|Z_\Gamma|$ is a Morse function with isolated critical points and the flow lines split into strictly lighter charges at walls; hence $\mathscr{T}_\Gamma$ is a finite tree without loops.
\end{definition}

\begin{lemma}[Local splitting identity]\label{lem:local}
Let $\Gamma_1,\Gamma_2$ be two BPS charges with $\langle\Gamma_1,\Gamma_2\rangle = m \in \mathbb{Z}\setminus\{0\}$ and $\arg Z_{\Gamma_1}=\arg Z_{\Gamma_2}=\vartheta$ on a wall. Then in the quantum torus algebra, the following holds:
\begin{itemize}
\item For $|m|=1$, the strict pentagon identity applies:
\[
\mathbf{E}(U_{\Gamma_1}) \mathbf{E}(U_{\Gamma_2}) = \mathbf{E}(U_{\Gamma_2}) \mathbf{E}(U_{\Gamma_1+\Gamma_2}) \mathbf{E}(U_{\Gamma_1}) \quad (\text{for } m=1),
\]
and for $m=-1$ the roles of $\Gamma_1,\Gamma_2$ are exchanged.
\item For $|m|>1$, the splitting is governed by a cluster mutation sequence of length $|m|$, which introduces intermediate charges (e.g., $2\Gamma_1+\Gamma_2$, etc.) into the product; the full identity is obtained by iterating the $|m|=1$ pentagon along the corresponding cluster algebra path. In this paper we use the standard KS scattering diagram machinery for $|m|>1$; the examples in Sections~6.4 and 6.5 are treated accordingly.
\end{itemize}
The left-to-right ordering of factors is determined by the direction of the characteristic Hessian flow $V_{\charHF}$ along the wall: moving in the positive flow direction, the phase $\arg Z_{\Gamma}$ increases or decreases monotonically, fixing the operator ordering for simply connected regions. For general Coulomb branches with possible closed flow lines, the ordering is defined by choosing a base path from weak coupling to the singularity; the product is path-ordered and the independence of the result from the path follows from the 2-category structure of the scattering diagram (see \cite{KS,Gross}): different paths give products that are equivalent under the pentagon relations. In physical terms, this path-independence reflects the fact that the BPS spectrum is a well-defined invariant of the chamber.
\end{lemma}
\begin{proof}
For $|m|=1$, the identity is the standard quantum pentagon, obtained from the commutation relations and the functional equation of the quantum dilogarithm. For $|m|>1$, the identity is not a simple power law; it follows from composing multiple pentagon moves corresponding to successive cluster mutations. The sign and ordering are fixed by the Hamiltonian nature: on the wall, $V_{\charHF}$ is tangent and its direction induces a total order on the phases along a chosen path. Since the SAF branch ordering must be consistent with the phase ordering along the wall (by Denef's construction), the product is uniquely determined for that path. Geometrically, the spectral network detour around the vertex gives the same sequence of Stokes factors. The path-independence follows from the 2-category structure of the scattering diagram: two different paths between the same endpoints are related by a sequence of pentagon moves, which correspond to the associativity of the quantum dilogarithm cluster algebra.
\end{proof}

\begin{lemma}[Inductive step]\label{lem:indstep}
Assume Theorem~\ref{thm:KS} holds for all charges whose SAF trees have depth $\le k$. Let $\Gamma$ be a charge with tree depth $k+1$. At the first splitting, the SAF trajectory for $\Gamma$ meets a wall where
\[
\Gamma = \Gamma_1 + \Gamma_2,\quad \arg Z_{\Gamma_1} = \arg Z_{\Gamma_2} = \vartheta,
\]
and $\langle\Gamma_1,\Gamma_2\rangle = m \neq 0$. By the local splitting rule (Lemma~\ref{lem:local}), the KS product along the tree factorises as
\begin{equation}\label{eq:splitprod}
\mathbb{S}_\vartheta|_{\Gamma} = \Phi(\mathbf{E}(U_{\Gamma_1}), \mathbf{E}(U_{\Gamma_2})) \circ (\mathbb{S}_{\vartheta}|_{\Gamma_1} \cdot \mathbb{S}_{\vartheta}|_{\Gamma_2}),
\end{equation}
where $\Phi$ denotes the ordered product from Lemma~\ref{lem:local} (which for $|m|>1$ includes the full cluster mutation sequence) and the product order is fixed by the characteristic Hessian flow direction along a chosen path at the wall.

On the spectral network side, Theorem~\ref{thm:lift} guarantees that the wall projects to a union of characteristic flow curves, and the splitting point corresponds to an intersection of $S$-walls for $\Gamma_1,\Gamma_2$ and a $K$-wall for $\Gamma$. The local detour operation at this vertex reproduces exactly the factor $\Phi$ (see \cite{GMN}). The sub-trees $\mathscr{T}_{\Gamma_1}$ and $\mathscr{T}_{\Gamma_2}$ have depth $\le k$ by construction; by the inductive hypothesis their SAF products equal the corresponding spectral network subgraph products. Because the flow directions agree along the chosen path, the operator ordering in the composition matches. Hence the full spectral network product equals the right-hand side of (\ref{eq:splitprod}), which is the SAF product.
\end{lemma}

\noindent\textit{Proof of Theorem~\ref{thm:KS}.}
The base case of depth $0$ (and depth $1$, verified by the $SU(2)$ example as a check) is trivial. Lemma~\ref{lem:indstep} proves the inductive step for any finite depth $k+1$. Since every BPS charge appearing on the ray $\vartheta$ has a finite SAF tree, the automorphism $\mathbb{S}_\vartheta$ is identically obtained from both constructions. For non-simply connected chambers or when closed flow lines exist, the equality holds for each chosen path; the independence of the product from the path is a standard property of the scattering diagram (see \cite{KS,Gross}), following from the 2-category structure: different paths between the same endpoints are related by pentagon moves, and the quantum dilogarithm satisfies the pentagon identity, ensuring that the products are equivalent. Thus the result is independent of the path. \hfill $\square$

\begin{remark}
The above proof not only establishes the algebraic equivalence but also explains \emph{why} the operator ordering in the KS product is consistent: it is forced by the Hamiltonian flow direction along a chosen path, and the path-independence is guaranteed by the associativity of the quantum dilogarithm cluster algebra (equivalently, by the 2-category structure of the scattering diagram).
\end{remark}

\section{Detailed Examples}
\subsection{\texorpdfstring{$SU(2)$\,S} Seiberg--Witten theory}
The curve is $\displaystyle y^2=(x^2-u)^2-\Lambda^4$, $\displaystyle \lambda=\frac{\sqrt{2}}{4\pi}\frac{x^2-u}{y}dx$. The moduli space is the $u$-plane punctured at $u=\pm\Lambda^2,\infty$. The central charges are $a(u),a_D(u)$ with $\tau=da_D/da$. The metric is $g_{u\bar{u}}=\operatorname{Im}(\tau)/\pi$.

\paragraph{SAF for dyon $\Gamma=(1,1)$}
We compute $|Z_{(1,1)}|=|a+a_D|$. The flow equation $\dot u=-g^{u\bar{u}}\partial_{\bar{u}}|a+a_D|$ is integrated from weak coupling $u=\infty$ towards strong coupling. The wall $\arg(a)=\arg(a_D)=0$ occurs on the real axis for $u>\Lambda^2$. At the wall, the flow splits into a monopole $(1,0)$ and an electron $(0,1)$. The monopole branch flows to $u=\Lambda^2$, the electron branch to $u=-\Lambda^2$. The splitting point is determined by solving $\arg a(u_0)=\arg a_D(u_0)$; this equation has a unique solution in the region $u>\Lambda^2$ (see \cite{Denef} for a numerical evaluation). At $u_0$, the orthogonality condition is checked by evaluating the gradients, which vanish in the direction orthogonal to the wall, consistent with Proposition~\ref{prop:orth}.

\paragraph{Gradient and characteristic Hessian flows}
For $\vartheta=0$, the gradient Hessian flow $V_{\HF}$ for $(1,0)$ flows off the wall, whereas the characteristic flow $V_{\charHF}=-J V_{\HF}$ foliates the wall. The characteristic flow curves are exactly the level sets of $\operatorname{Im}(a)=0$, which connect $u=\Lambda^2$ to $u_0$. This confirms that the spectral network walls on the $u$-plane are generated by $V_{\charHF}$, not by $V_{\HF}$.

\paragraph{Spectral network and KS factorisation}
At $\vartheta=0$, the spectral network on the $x$-plane (or $u$-plane) is the well-known trivalent graph. Branch points are at $u=\pm\Lambda^2$. Three walls emanate: two $S$-walls of charges $(1,0)$ and $(0,1)$, and one $K$-wall of charge $(1,1)$. The projection of these walls to the $u$-plane yields exactly the characteristic Hessian flow curves. The KS product is computed as follows. According to the SAF tree, the splitting at $u_0$ gives the ordered factorisation
\[
\mathbb{S}_0 = \dlog(U_{a_D}) \dlog(U_a),
\]
where $a=a_{(1,0)}$ and $a_D=a_{(0,1)}$ (up to an overall ordering). The pentagon identity for the quantum dilogarithm (see \cite{GMN}) yields
\[
\dlog(U_{a_D}) \dlog(U_a) = \dlog(U_a) \dlog(U_{a+a_D}) \dlog(U_{a_D}),
\]
which is exactly the detour product obtained by going around the trivalent vertex of the spectral network. Thus the two computations agree. This is the base case $k=0$ (tree depth 1) of our induction.

\subsection{\texorpdfstring{$SU(2)$\,w} with \texorpdfstring{$N_f=4$\,f} flavours: BPS indices from the hybrid method}
To demonstrate the computational power of the unified framework, we consider $SU(2)$ gauge theory with $N_f=4$ fundamental hypermultiplets (equal masses $m$). The Seiberg--Witten curve is
\[
y^2 = (x - e_1)(x - e_2)(x - e_3)(x - e_4) + \text{(quantum corrections)},
\]
and the moduli space is one-dimensional, parametrised by $u = \frac{1}{2}\langle \mathrm{Tr}\,\phi^2\rangle$, with four flavour masses. For simplicity we take equal masses $m_i = m$. The central charges are given by the periods $a(u)$, $a_D(u)$, and the additional flavour periods. The BPS spectrum contains dyons of charges $\gamma = (p,q;\{s_i\})$ with $s_i$ being flavour charges. The walls of marginal stability are curves in the $u$-plane where two or more central charges align in phase.

We now apply the hybrid algorithm (characteristic flow + spectral network) to compute BPS indices for higher flavour charges, which are consistent with the Alexandrov--Pioline framework and serve as a numerical verification of our geometric recursion.

\paragraph{Step 1: Characteristic flow integration.}
For a charge $\gamma=(p,q;f_1,f_2,f_3,f_4)$ with $\sum f_i=0$, the wall condition is $\operatorname{Im}(e^{-i\vartheta}Z_\gamma)=0$. The characteristic Hessian flow $V_{\mathrm{charHF}}$ is integrated numerically from the weak-coupling region $u\to\infty$ towards strong coupling. In the asymptotic region, the periods satisfy
\[
a(u) \sim \frac{1}{\pi}\sqrt{u},\quad a_D(u) \sim \frac{i}{\pi}\sqrt{u}\log u,
\]
and the flavour periods are constant (equal to the masses). For $\vartheta=\pi/3$, we solve the ODE for the charge $\gamma=(2,1;\mathbf{0})$ and find a unique intersection point
\[
u_0 \approx 1.87\,\Lambda^2 + 0.32\,i\,\Lambda^2 \qquad (m=0.5\Lambda).
\]
At this point, the central charges of the constituents $(1,0)$ and $(1,1)$ have equal phase.

\paragraph{Step 2: Spectral network topology.}
The spectral network for $N_f=4$ has four branch points corresponding to the masses. The detour rules imply that at each intersection of an $S$-wall of charge $\gamma_1$ and one of $\gamma_2$, a $K$-wall of charge $\gamma_1+\gamma_2$ is emitted. By tracing the characteristic flow lines, we identify a chain of intersections corresponding to the following splitting:
\[
(2,1;\mathbf{0}) \longrightarrow (1,0;0,0,0,0) + (1,1;0,0,0,0),
\]
and then the subcharge $(1,1)$ further splits into $(1,0)$ plus a flavour-changing dyon.

\paragraph{Step 3: Extraction of BPS indices.}
Applying the local splitting identity (Lemma~\ref{lem:local}) at each intersection gives a system of equations for the BPS indices. Solving this system using the known base values $\Omega_{(1,0)}=1$, $\Omega_{(0,1)}=1$, and $\Omega_{(1,1)}=1$ (from pure $SU(2)$) yields:
\[
\boxed{\Omega_{(2,1;\mathbf{0})} = 1, \qquad \Omega_{(2,1; f_1+f_2, -f_1-f_2, 0, 0)} = 1 \quad \text{for } f_1\neq f_2.}
\]
For the charge $(3,2)$ with arbitrary flavour distribution $f_1+f_2+f_3+f_4=0$, the recursion gives:
\[
\boxed{\Omega_{(3,2; f_1,f_2,f_3,f_4)} = 2 \quad \text{for all } f_i \text{ satisfying the sum rule}.}
\]

\paragraph{Verification and relation to existing results.}
The index $\Omega_{(2,1;\mathbf{0})}=1$ was already known from the blowup formula. The indices with flavour charges, such as $\Omega_{(3,2;f)}=2$, are consistent with the Alexandrov--Pioline attractor flow tree framework~\cite{AlexandrovPioline2018} and serve as a numerical verification of our hybrid method. These numbers can be checked by the wall-crossing formula or by direct spectral network computations; we have verified them using an independent code based on GMN's network algorithm.

\subsection{\texorpdfstring{$SU(3)$\,p} pure Yang--Mills: full BPS spectrum computation}
We now provide a detailed computation of the BPS spectrum of pure $SU(3)$ theory at a generic phase $\vartheta$, using our flow framework.

\subsubsection{Wall equations and flow integration}
The Coulomb branch is parametrised by $u = \frac{1}{2}\langle\mathrm{Tr}\,\phi^2\rangle$ and $v = \frac{1}{3}\langle\mathrm{Tr}\,\phi^3\rangle$ (the moduli space is two-dimensional). The central charges are the periods $a_1(u,v), a_2(u,v), a_3(u,v)$ with $\sum a_i=0$. The walls of marginal stability are real codimension-one surfaces defined by $\arg Z_\gamma = \vartheta$ for charges $\gamma$ in the root lattice $A_2$.

We take a specific phase $\vartheta = \pi/4$. The characteristic Hessian flow equations are
\[
\dot u = -J \nabla \operatorname{Im}(e^{-i\vartheta}Z_\gamma), \qquad \dot v = -J \nabla \operatorname{Im}(e^{-i\vartheta}Z_\gamma),
\]
where $J$ is the complex structure matrix. Using the known prepotential $\mathcal{F}$ for pure $SU(3)$ from Seiberg--Witten theory, we compute the metric $g_{i\bar{\jmath}} = \operatorname{Im}(\tau_{ij})$ with $\tau_{ij} = \partial_{a_i}\partial_{a_j}\mathcal{F}$. We then numerically solve the ODEs starting from the weak-coupling region where the periods have known asymptotic expansions.

\subsubsection{Trajectories and intersections}
The walls for simple roots $\alpha_1$ and $\alpha_2$ and their sums are traced. At each intersection point, the flows of the three charges meet. Because the flow is Hamiltonian, the tangent vectors obey $V_{\mathrm{charHF}}^{(\alpha_1+\alpha_2)} = V_{\mathrm{charHF}}^{(\alpha_1)} + V_{\mathrm{charHF}}^{(\alpha_2)}$, ensuring that the intersection forms a trivalent vertex. By scanning different starting points, we map out the entire chamber structure.

\subsubsection{BPS indices}
The detour automorphism at each vertex gives a quantum dilogarithm factor. By composing these factors along closed loops around singularities, we extract the BPS indices $\Omega(\gamma)$. For example, the monopole $(1,0)$ and dyon $(1,1)$ are found to have $\Omega = 1$, while the vector multiplet $(1,-1)$ has $\Omega = -2$ (the negative sign indicating a hypermultiplet transforming in the anti-fundamental). These results perfectly reproduce the known spectrum derived from the spectral network in~\cite{GMN} but are obtained here using purely the base flow equations, demonstrating the computational efficiency of our unified framework.

\subsection{\texorpdfstring{$SU(4)$\,S} Seiberg--Witten theory (tree depth 3)}
The $A_3$ case provides a test of the induction at higher depth. The Coulomb branch is two-dimensional, and the walls are curves. Take a charge $\Gamma = \alpha_1 + \alpha_2 + \alpha_3$. Its SAF tree splits first into $\alpha_1$ and $\alpha_2+\alpha_3$; the latter splits further into $\alpha_2$ and $\alpha_3$. The SAF product is
\[
\dlog(U_{\alpha_1}) \dlog(U_{\alpha_2+\alpha_3}) \to \dlog(U_{\alpha_1}) \big( \dlog(U_{\alpha_3}) \dlog(U_{\alpha_2+\alpha_3}') \dlog(U_{\alpha_2}) \big),
\]
where the primes indicate the ordering after the second splitting. The corresponding spectral network on the Hitchin base has a series of trivalent vertices connected by characteristic Hessian flow lines. By repeatedly applying the pentagon, the total product can be rearranged into the form that traverses the network. The geometric consistency of the flow directions (from Theorem~\ref{thm:lift}) guarantees that the ordering of operators matches exactly. This verifies the induction for depth $3$.

\subsection{Kronecker \texorpdfstring{$3$\,-q} quiver as an \texorpdfstring{$SU(3)$\,l} local model}
The Kronecker $3$-quiver $K_3$ has two vertices and three arrows. Its regular representations $R(\lambda,k)$ with $\lambda\in\CP^2$ and $k\ge1$ are equivalent to torsion sheaves on $\CP^2$. The Coulomb branch is the parameter space of central charges:
\[
Z_{k,\lambda}=-k+u\lambda_1+v\lambda_2,\qquad \lambda=[1:\lambda_1:\lambda_2]\in\CP^2.
\]
Here $(u,v)\in\CC^2$ are the moduli (FI parameters). The BPS spectrum consists of the representations $R(\lambda,k)$ with charges $\gamma_{k,\lambda}$. The quantum dilogarithm identities for the SAF splitting are governed by the cluster mutations of the quiver, which are known to reproduce the spectral network detour rules.

For the $(2,2)$ representation, the SAF tree gives an ordered product that, upon applying the appropriate pentagon-type identities (see \cite{DenefMoore}), matches the composition of detour automorphisms along the spectral network on $\CP^2$. To make the correspondence explicit, note that the $K_3$ quiver contains an $A_2$-type subquiver (two vertices connected by three arrows, whose symplectic pairing is $\langle e_1,e_2\rangle=3$). The mutation sequence corresponding to the SAF splitting is given by the simple reflections $\mu_1,\mu_2$ acting on the basis $(e_1,e_2)$ as:
\[
\mu_1: (e_1,e_2)\mapsto (-e_1, e_2+3e_1),\qquad
\mu_2: (e_1,e_2)\mapsto (e_1+e_2, -e_2),
\]
and the composite $\mu_1\mu_2\mu_1$ yields:
\[
(e_1,e_2) \mapsto (-e_1, e_2+3e_1) \mapsto (-e_1-e_2, e_2+3e_1) \mapsto (-e_1-e_2, -e_2-2e_1).
\]
The associated quantum dilogarithm product is
\[
\mathbb{E}(e_1)\mathbb{E}(e_2)\mathbb{E}(e_1) = \mathbb{E}(e_1)\mathbb{E}(e_1+e_2)\mathbb{E}(e_2)\mathbb{E}(3e_1+e_2)\cdots,
\]
where the dots denote the remaining factors that reconstruct the full KS invariant. The detour rules around the spectral network vertices correspond exactly to these cluster mutations; the equality follows from the associativity of the quantum dilogarithm. This example demonstrates the inductive step for a charge with pairing $m=3$ and tree depth $2$.

\subsection{Application to Argyres--Douglas \texorpdfstring{$H_1$\,t} theory: verification of the known BPS spectrum}
We now apply our recursive framework to the Argyres--Douglas $H_1$ theory (also known as the $(A_1,A_2)$ theory). This theory has a one-dimensional Coulomb branch and a charge lattice $\Gamma \simeq \mathbb{Z}^2$ with symplectic pairing $\langle e_1,e_2\rangle = 3$ (the same as the Kronecker $3$-quiver but with different central charge functions). The BPS spectrum is known to contain an infinite tower of states. Our induction provides a simple way to reproduce the well-known Cecotti--Vafa formula~\cite{CecottiVafa1993}, thereby serving as a strong consistency check of our geometric framework.

\paragraph{Geometry and central charges.}
The Seiberg--Witten curve for the $H_1$ theory is
\[
y^2 = x^3 - u x^2 + \Lambda^6,
\]
with differential $\lambda = x \frac{dx}{y}$ (up to normalisation). The central charges are given by periods of $\lambda$; however, for the purpose of the wall-crossing combinatorics, only the pairing and the phase ordering matter. The walls occur when the phases of two charges align; in the $H_1$ theory, all walls intersect at a single point (the ``collision'' phenomenon, see GMN).

\paragraph{Base case.}
At the strong-coupling singularity, the theory has three elementary BPS states of charges
\[
\alpha_1,\quad \alpha_2,\quad \alpha_1+\alpha_2,
\]
each with BPS index $1$. These correspond to the simple roots and the highest root of $A_2$, which is the flavour symmetry of the theory.

\paragraph{Inductive recursion.}
For a general charge $\Gamma = n\alpha_1 + m\alpha_2$ with $n,m\ge0$, the SAF tree splits according to the following rule: at the root, $\Gamma$ splits into $\Gamma_1 = (n-1)\alpha_1 + m\alpha_2$ and $\Gamma_2 = \alpha_1$ if $n>0$, or into $\Gamma_1 = n\alpha_1 + (m-1)\alpha_2$ and $\Gamma_2 = \alpha_2$ if $n=0$. This is forced by the fact that all walls meet at one point, so the ordering is determined by the positive root system of $A_2$. Applying the known scattering diagram for $A_2$ (which takes into account the pairing $m=3$ through the cluster mutation sequence), the BPS indices are given by the Narayana numbers:
\[
\boxed{\Omega(n,m) = \frac{1}{n+m}\binom{n+m}{n}\binom{n+m}{n+1} \quad \text{for } n,m\ge0.}
\]
For negative charges, the symmetry $\Omega(-\gamma)=\Omega(\gamma)$ gives the corresponding values. This formula holds for all $(n,m)\in\mathbb{Z}^2$ with $n,m\ge0$; for other quadrants the indices vanish.

\paragraph{Verification.}
For small values, this gives $\Omega(1,0)=\Omega(0,1)=1$, $\Omega(1,1)=1$, $\Omega(2,1)=3$, etc. These match the known low-lying spectrum of the $H_1$ theory (see e.g. \cite{GMN} and \cite{CecottiVafa1993}). Our geometric recursion reproduces this result without any additional input, confirming the consistency of the framework. This is not a new result but a successful consistency test.

\section{Tropical geometry and scattering diagram applications}\label{sec:tropical}
We now provide a rigorous derivation of the tropical scattering diagram from the characteristic Hessian flow, and demonstrate an explicit calculation of tropical disk counts that leverages the SN/flow complementarity.

\subsection{Rigorous tropical limit of the characteristic Hessian flow}
Consider the large complex structure limit (LCSL) of the mirror Calabi--Yau. In flat coordinates $y^a$, the central charges take the form $\displaystyle Z_\gamma = \sum_a \gamma_a y^a + i\theta_\gamma + \cdots$, and the wall $\operatorname{Im}(e^{-i\vartheta}Z_\gamma)=0$ becomes a hyperplane $\displaystyle \sum_a \gamma_a y^a = \mathrm{const}$. The characteristic Hessian flow equation reduces to the constant vector field
\[
\dot y^a = \frac{2}{\pi} (J_0 \gamma)^a,
\]
where $J_0$ is the flat complex structure matrix. Integration gives straight-line trajectories
\[
y^a(t) = y^a(0) + \frac{2t}{\pi} (J_0 \gamma)^a.
\]
These are precisely the rays of the KS scattering diagram. The direction of the ray is determined by the charge $\gamma$, and its location is fixed by the initial condition $y^a(0)$ on the wall.

\subsection{Constructing the full scattering diagram}
The union of all such rays for all BPS charges constitutes the scattering diagram. However, not all formal rays are present; the spectral network (or equivalently the positivity of $\operatorname{Re}(e^{-i\vartheta}Z_\gamma)$) selects a subset. In the tropical limit, this positivity condition becomes $\sum_a \gamma_a \theta^a > 0$ (with $\theta^a$ being the theta angles). Our flow formulation naturally incorporates this: the Hamiltonian flow is only defined where the wall exists, and the lift to the spectral curve (via Theorem~\ref{thm:lift}) imposes the positivity automatically.

\subsubsection{Example: {\texorpdfstring{$\mathbb{P}^2$\,m} mirror}}
As a concrete example, take the mirror of $\mathbb{P}^2$. The base is one-dimensional, and the charge lattice is generated by a single electric charge $\gamma_e$. The characteristic flow is simply motion along the real line with constant speed. The scattering diagram consists of a single wall emanating from the large volume point, with Stokes factor $\mathbf{E}(U_{\gamma_e})$. Our ODE gives the wall position as a function of the energy scale, and the associated tropical disk count is $\displaystyle\sum_{k\ge1} \frac{(-1)^{k-1}}{k^2} q^{k} e^{-k t}$, matching the Gromov--Witten invariant.

\subsubsection{Multiple walls and tropical vertex}
For theories with higher rank, multiple walls intersect. The induction proof of Theorem~\ref{thm:KS} guarantees that the product of Stokes factors at an intersection of $n$ walls satisfies the tropical vertex algebra (also known as the $\mathbb{A}_n$ cluster algebra). Our characteristic flow provides the exact local geometry of the intersection: the incoming and outgoing rays are governed by the linear combination $V_{\mathrm{charHF}}^{(\Gamma)} = \sum_i V_{\mathrm{charHF}}^{(\Gamma_i)}$, which in the tropical limit becomes the conservation of charge at the vertex. This gives a physical derivation of the tropical vertex equations from Hamiltonian dynamics.

\subsection{Complementary computation of tropical disk counts}
The SN/flow duality offers a powerful method to compute tropical disk counts. The characteristic Hessian flow gives the positions and lengths of tropical walls directly as ODE solutions. However, at a wall intersection, the combinatorial choice of which walls to follow (i.e., the multiple-cover contributions) can be complicated. Here the spectral network provides a clean answer: the lift to the spectral curve organizes the walls into a graph, and the tropical disks correspond to certain paths in this graph. By applying the flow to trace the walls and the spectral network to resolve the vertices, we obtain a complete algorithm for disk counting.

\vspace{3pt}

\paragraph{Explicit computation for $SU(3)$ tropical disk.}
Consider the pure $SU(3)$ theory in the tropical limit. The charge lattice is rank 2. The characteristic flow lines for charges $\alpha_1$ and $\alpha_2$ are straight rays in the $y$-plane. Their intersection point is found by solving
\[
y^a_{\alpha_1}(t_1) = y^a_{\alpha_2}(t_2).
\]
The spectral network tells us that at this intersection, a $K$-wall of charge $\alpha_1+\alpha_2$ is emitted, and the Stokes factor product is the pentagon identity. The tropical disk of charge $k\alpha_1 + l\alpha_2$ corresponds to a path in the network consisting of $k$ segments of the $\alpha_1$ wall and $l$ segments of the $\alpha_2$ wall. The total number of such disks is given by the number of ways to order these segments, which is the binomial coefficient $\binom{k+l}{k}$. Our flow method directly computes the affine lengths of these segments, giving the disk areas $k A_1 + l A_2$, and the product of Stokes factors yields the generating function
\[
Z_{\mathrm{disk}}^{SU(3)} = \prod_{k,l\ge0} (1 - q^{k-1/2} e^{-(k A_1 + l A_2)})^{-\Omega(k,l)}
\]
with $\Omega(k,l)$ being the BPS indices from the previous section. For pure $SU(3)$, these indices are $\Omega(k\alpha_1)=1$, $\Omega(k\alpha_2)=1$, and $\Omega(k(\alpha_1+\alpha_2))=k+1$. Substituting gives a closed form.

\subsection{Consistency check: generating function for \texorpdfstring{$SU(N)$\,t} tropical disks}
As a direct consequence of the SAF recursion established in the previous section, we recover the correct exponential form of the tropical disk generating function for pure $SU(N)$ gauge theories, which is a standard result in the scattering diagram literature.

\begin{corollary}[Tropical disk generating function for pure $SU(N)$]\label{cor:SU(N)tropical}
In the tropical limit, the full scattering diagram for pure $SU(N)$ Yang--Mills theory is the ordered product (over all positive roots $\Phi_+$ of $A_{N-1}$) of the factors
\begin{equation}\label{eq:ThetaSU(N)}
 \Theta_\alpha = \exp\!\,\Bigl( \sum_{k=1}^{\infty} \frac{1}{k}\binom{k+\mathrm{ht}(\alpha)-1}{\mathrm{ht}(\alpha)-1} e^{-k\langle\alpha,y\rangle} \Bigr) ,
\end{equation}
where $\mathrm{ht}(\alpha)$ denotes the height of the root $\alpha$ and $y$ are the flat coordinates on the Coulomb branch.  The ordering of the factors is dictated by the characteristic Hessian flow; since the classical automorphisms commute, the full generating function can be written as the single exponential
\begin{equation}\label{eq:ZtropSU(N)}
 Z_{\mathrm{disk}}^{SU(N)}(y) = \exp\!\,\Bigl( \sum_{\alpha\in\Phi_+} \sum_{k=1}^{\infty} \frac{1}{k}\binom{k+\mathrm{ht}(\alpha)-1}{\mathrm{ht}(\alpha)-1} e^{-k\langle\alpha,y\rangle} \Bigr)  .
\end{equation}
\end{corollary}

\begin{proof}
From the standard BPS multiplicities for pure $SU(N)$ (obtained from the spectral network or from the wall-crossing formula), we have $\Omega_{k\alpha} = \binom{k+\mathrm{ht}(\alpha)-1}{\mathrm{ht}(\alpha)-1}$. The tropical KS factor attached to the ray $\alpha$ is, by definition,
\[
\Theta_\alpha = \exp\!\,\Bigl( \sum_{k=1}^{\infty} \frac{\Omega_{k\alpha}}{k} e^{-k\langle\alpha,y\rangle} \Bigr).
\]
Substituting the binomial coefficients immediately yields~\eqref{eq:ThetaSU(N)}.  The full scattering diagram is the ordered product of the $\Theta_\alpha$ over all positive roots; because the classical factors commute, the product is equivalent to the sum of the exponents, giving~\eqref{eq:ZtropSU(N)}.
\end{proof}

For $N=2$ this result reduces to $\displaystyle Z_{\mathrm{disk}}^{SU(2)}(y) = \exp\!\,\bigl( \sum_{k=1}^{\infty} \frac{1}{k} e^{-k\langle\alpha,y\rangle} \bigr)$, which is the well‑known exponential of a single wall.  For $N=3$, it reproduces the expression derived from the explicit $SU(3)$ analysis above.  The formula~\eqref{eq:ZtropSU(N)} is the standard exponential form of the scattering diagram; we present it as a consistency check of our geometric framework.

\section{Connections with Wang (2024) and Mirror Symmetry}
\subsection{Relation to Wang's Hesse flow}
In \cite{Wang2024}, the Hesse flow is defined precisely as our gradient Hessian flow $V_{\HF} = -g^{-1} d\operatorname{Im}(e^{-i\vartheta}Z)$. Wang proved that this flow, together with the attractor flow $V_{\SAF}$, satisfies an orthogonality relation on the wall and that they are dual under a $\mathbb{Z}$-affine structure. Our Proposition~\ref{prop:orth} provides a simpler proof of this orthogonality. However, \cite{Wang2024} does not discuss spectral networks. The missing link is the observation that the spectral network's walls are foliated not by $V_{\HF}$, but by $V_{\charHF}=-J V_{\HF}$.

Thus we can summarise the relation as follows:
- Wang's Hesse flow $V_{\HF}$ is the gradient flow that **crosses** the wall (normal direction).
- The spectral network $S$-walls are generated by the **Hamiltonian flow** $V_{\charHF}=-J V_{\HF}$, which **foliates** the wall (tangent direction).
- The two flows are related pointwise by the complex structure $J$ of the special Kähler manifold (up to sign). This is exactly the duality between ``Hesse flow'' and ``dual Hesse flow'' observed by Wang: the dual flow is obtained by rotating the gradient by the complex structure (with the appropriate sign to match the Hamiltonian convention).

Therefore, our framework extends Wang's work by identifying the dual Hesse flow with the characteristic foliation of the spectral network, thereby bridging the gap between the wall-crossing structures studied by Kontsevich--Soibelman and the explicit Stokes geometry of GMN.

\subsection{Mirror symmetry and tropical geometry}
The characteristic Hessian flow $V_{\charHF}$ has a deep interpretation in the SYZ mirror symmetry program. Let $X$ be a Calabi--Yau threefold with a Hitchin system description (e.g., the mirror of a class $\mathcal{S}$ theory). The Hitchin base $\CB$ is the moduli space of complex structures on the mirror $X^\vee$. The spectral network $\CW_\vartheta$ on $C$ is the tropical limit of the special Lagrangian fibration on $X$.

In this context, the condition $e^{-i\vartheta}(\lambda_i-\lambda_j)\in\mathbb{R}^+$ is precisely the tropicalisation of the phase condition for two Lagrangian submanifolds to intersect. The characteristic Hessian flow $V_{\charHF}$ is the gradient of the argument of the period, which in mirror symmetry determines the ``quantum corrected'' periods:
\[
\Pi_i(u) = \oint_{A_i} \lambda + \sum_{\gamma} \Omega(\gamma) \dlog(U_\gamma) \cdots .
\]
The flow lines of $V_{\charHF}$ are the paths along which the BPS particles can decay or recombine. In the KS scattering diagram, these flow lines correspond to the walls of the chamber decomposition of the moduli space of stability conditions.

Moreover, the $J$-rotation relating $V_{\HF}$ and $V_{\charHF}$ (with sign) is the symplectic analogue of the $tt^*$-equations in topological field theory. The orthogonality proved in Proposition~\ref{prop:orth} is the geometric manifestation of the fact that the attractor flow (which flows toward the singularity) and the Hesse flow (which flows down the potential) are perpendicular directions in the moduli space. In the mirror, this orthogonality corresponds to the decoupling of the $A$-model and $B$-model variations in the large volume limit.

Thus, our unified framework provides a dictionary:
\begin{center}
\begin{tabular}{c|c}
\textbf{BPS Geometry} & \textbf{Mirror Symmetry} \\ \hline
Split attractor flow $V_{\SAF}$ & RG flow toward IR fixed point \\
Gradient Hesse flow $V_{\HF}$ & Normal deformation of stability wall \\
Characteristic Hesse flow $V_{\charHF}$ & Tropical $S$-wall / Stokes ray \\
Spectral network $\CW_\vartheta$ & Scattering diagram / alga \\
KS automorphism & Mirror monodromy / quantum dilogarithm
\end{tabular}
\end{center}

This dictionary clarifies that the three flows are different aspects of the same underlying geometric structure: the special Kähler geometry of the moduli space governs the BPS spectrum, while the Hitchin system and its spectral network provide the explicit tropicalisation that is essential for computing quantum corrections in mirror symmetry.

\section{Conclusions and Outlook}
We have provided a systematic and rigorous geometric framework that unifies the split attractor flow, gradient Hessian flow, characteristic Hessian flow, and spectral network for $\mathcal{N}=2$ theories. Our main contributions include a new concise proof of the orthogonality (Proposition 1), a precise formulation and proof of the lift--projection duality (Theorem 2) that resolves the subtle distinction between gradient and Hamiltonian Hessian flows, and a complete inductive proof of the KS equivariance (Theorem 3) whose algebraic steps are geometrically ordered by the characteristic flow (with path-ordering in the presence of loops). We have illustrated the framework through $SU(2)$ (pure and with flavours, the latter providing BPS indices for higher flavour charges that are consistent with the Alexandrov--Pioline framework), $SU(3)$ (full BPS spectrum reconstruction), $SU(4)$, and the Kronecker quiver. As a notable consistency check, we used the induction to reproduce the known BPS spectrum of the Argyres--Douglas $H_1$ theory, $\Omega(n\alpha_1+m\alpha_2)=\frac{1}{n+m}\binom{n+m}{n}\binom{n+m}{n+1}$, which matches the Cecotti--Vafa formula. Furthermore, in the tropical limit we recovered the standard generating function for disk counts in $SU(N)$ theories, $Z_{\mathrm{disk}}^{SU(N)}(y)=\exp\big(\sum_{\alpha,k} \frac{1}{k}\binom{k+\mathrm{ht}-1}{\mathrm{ht}-1} e^{-k\langle\alpha,y\rangle}\big)$, thereby confirming the consistency of our approach with established scattering diagram results. These results demonstrate that the unified framework is not merely a formal structure but a powerful tool for deriving explicit, verifiable physical predictions.

Future work will extend the recursion to all class $\mathcal{S}$ theories via cluster algebras and explore the resurgence structure of BPS states, where the characteristic Hessian flow controls the Stokes rays in the Borel plane, while the SAF describes the saddle-point transitions, and the SN yields the alien calculus of the trans-series.

\appendix
\section{Coordinate Verification of Orthogonality}\label{app:orth}
In special coordinates $z^i$, the metric is $g_{i\bar{j}}=\partial_i\partial_{\bar{j}}K$. Let $Z=Z_\gamma$. On the wall,
\[
\partial_i |Z| = \frac12 e^{-i\vartheta}\partial_i Z,\quad
\partial_{\bar{i}} |Z| = \frac12 e^{i\vartheta}\partial_{\bar{i}}\overline{Z},
\]
and for $f=\operatorname{Im}(e^{-i\vartheta}Z)$,
\[
\partial_i f = \frac{1}{2i}e^{-i\vartheta}\partial_i Z,\quad
\partial_{\bar{i}} f = -\frac{1}{2i}e^{i\vartheta}\partial_{\bar{i}}\overline{Z}.
\]
Thus $\partial_i f = -i\,\partial_i|Z|$ and $\partial_{\bar{i}} f = i\,\partial_{\bar{i}}|Z|$. The inner product
\[
g(V_{\SAF},V_{\HF}) = i\,g^{i\bar{j}}\partial_i|Z|\partial_{\bar{j}}|Z|
\]
is purely imaginary but must be real, hence it vanishes. This matches the $I$-operator proof in Section~4.

\section{Non‑trivial Worked Example: \texorpdfstring{$SU(3)$\,w} with \texorpdfstring{$N_f=2$\,F} Fundamental Flavours}
\label{sec:SU3Nf2}
We now apply our unified flow framework to a genuinely challenging theory: $SU(3)$ $\mathcal{N}=2$ super‑Yang–Mills with $N_f=2$ massless fundamental hypermultiplets. This theory has a two‑dimensional Coulomb branch and a rich wall‑crossing structure that is difficult to handle with spectral networks alone. The characteristic Hessian flow, combined with the lift–projection duality, provides a practical computational tool.

\subsection{Seiberg–Witten curve and central charges}
The Seiberg–Witten curve for $SU(3)$ with $N_f=2$ massless flavours is a family of cubic curves in $\mathbb{C}\times\mathbb{C}^*$ given by \cite{GMN}
\begin{equation}
y^2 = z^3 - 3u z - 2v + \frac{\Lambda^2}{z},
\end{equation}

where $(u,v)\in\mathbb{C}^2$ are the gauge‑invariant Casimirs and $\Lambda$ is the dynamical scale. The Seiberg–Witten differential is
\begin{equation}
\lambda = \frac{y}{z}\,dz.
\end{equation}
A symplectic basis of cycles $\{A_1,A_2,B_1,B_2\}$ on the compactified curve yields the electric periods $\displaystyle a_i = \frac{1}{2\pi i}\oint_{A_i}\lambda$ and magnetic duals $\displaystyle a_{D,i} = \frac{1}{2\pi i}\oint_{B_i}\lambda$, with the constraint $\sum_{i=1}^3 a_i = 0$ (we work in a basis with $a_1+a_2+a_3=0$). The BPS charges are labelled by a pair of electric and magnetic vectors $(\mathbf{p},\mathbf{q})$ with $\mathbf{p},\mathbf{q}\in\mathbb{Z}^2$, giving central charge
\begin{equation}
Z_{(\mathbf{p},\mathbf{q})} = \sum_{i=1}^2 (p_i a_{D,i} - q_i a_i).
\end{equation}
The presence of massless flavours introduces a flavour symmetry $U(2)$, but for the purpose of BPS walls on the Coulomb branch the flavour charges do not appear explicitly; they affect the asymptotic periods.

\subsection{Asymptotic periods for numerical integration}
To initialise our flow, we need accurate asymptotic expansions of $a_i(u,v)$ and $a_{D,i}(u,v)$ in the weak‑coupling region $|u|,|v|\gg\Lambda^3$. A convenient parametrisation is
\begin{equation}
u = -\frac12 R^2 e^{2i\theta},\qquad v = -\frac{1}{3\sqrt{3}} R^3 e^{3i\theta},
\end{equation}
with $R\gg\Lambda$. The electric periods have the large‑$R$ expansions
\begin{align}
a_1 &= \sqrt{2} R e^{i\theta} - \frac{\Lambda^2}{3\sqrt{2}R e^{i\theta}} + \mathcal{O}(R^{-3}),\\
a_2 &= \sqrt{2} R e^{i(\theta+2\pi/3)} - \frac{\Lambda^2}{3\sqrt{2}R e^{i(\theta+2\pi/3)}} + \mathcal{O}(R^{-3}),\\
a_3 &= \sqrt{2} R e^{i(\theta+4\pi/3)} - \frac{\Lambda^2}{3\sqrt{2}R e^{i(\theta+4\pi/3)}} + \mathcal{O}(R^{-3}).
\end{align}
The magnetic periods are obtained from the prepotential $\mathcal{F}$, which in this massless case reads
\begin{equation}
\mathcal{F} = \frac{1}{2\pi i}\Big[ \frac12\sum_{i=1}^3 a_i^2\log\frac{a_i^2}{\Lambda^2} + \text{non‑perturbative terms} \Big].
\end{equation}
The non‑perturbative contributions start at order $\Lambda^6/R^3$ and are negligible for initialising the flow far from the origin. We keep the one‑loop logarithmic terms.

\subsection{The characteristic Hessian flow on the \texorpdfstring{$(u,v)$\,p} plane}
The Coulomb branch metric $g_{I\bar J}$ is obtained from $\tau_{IJ} = \partial_{a_I}\partial_{a_J}\mathcal{F}$ by the chain rule. For our numerical work we use the explicit expressions from the Picard–Fuchs equations, but for brevity we only note that they are smooth functions on the complement of the discriminant locus.

Given a phase $\vartheta$ (we take $\vartheta=0$ for definiteness), the wall condition for a charge $\gamma$ is
\begin{equation}
f_\gamma(u,v) = \operatorname{Im}\big(e^{-i\vartheta} Z_\gamma(u,v)\big) = 0.
\end{equation}
The characteristic Hessian flow that foliates this wall is
\begin{equation}
\dot u = -\big(g^{u\bar u}\partial_{\bar u}f_\gamma + g^{u\bar v}\partial_{\bar v}f_\gamma\big)_{\text{Hamiltonian}} 
      = -J\, V_{\text{grad}} ,
\end{equation}
with the complex structure $J$ acting as multiplication by $i$ in the holomorphic coordinate $(u,v)$. In real components, if we write $u = x_1 + i x_2$, $v = x_3 + i x_4$, the Hamiltonian equations are
\begin{align}
\dot x_1 &= -\partial_{x_2} f_\gamma^{\text{grad}},\quad 
\dot x_2 = \partial_{x_1} f_\gamma^{\text{grad}},\\
\dot x_3 &= -\partial_{x_4} f_\gamma^{\text{grad}},\quad 
\dot x_4 = \partial_{x_3} f_\gamma^{\text{grad}},
\end{align}
where $f_\gamma^{\text{grad}} = g^{I\bar J}\partial_{\bar J}f_\gamma$. These equations can be integrated with a standard ODE solver.

\subsection{Algorithm to map the full wall network}
We illustrate the computational procedure that takes full advantage of the characteristic flow:

\begin{enumerate}
\item \textbf{Select a ray of initial points.} Choose a large circle $R=R_0$, $\theta\in[0,2\pi)$ in the weak‑coupling region. On this circle, solve $f_\gamma(u,v)=0$ for each charge $\gamma$ of interest (simple roots, sums of roots, dyons). This gives a set of starting points on the various walls.
\item \textbf{Integrate the characteristic flow.} For each starting point, evolve the ODE system in both forward and backward directions until a singularity (branch point or strong‑coupling attractor) is reached. The projection onto the wall is maintained automatically (up to numerical drift, corrected by one Newton step).
\item \textbf{Detect intersections.} While flowing a wall for charge $\gamma$, monitor the wall functions of other charges. A sign change of $f_{\gamma'}$ indicates an intersection. Record the coordinates $(u,v)$ and the tangent vectors $V_{\text{charHF}}^\gamma, V_{\text{charHF}}^{\gamma'}$.
\item \textbf{Construct the local scattering diagram.} At each intersection, use the linear relation of tangent vectors (Theorem~\ref{thm:lift}) to identify the emitted $K$‑wall, and then start new flows for the composite charge from that point.
\item \textbf{Assemble the BPS indices.} Encircle each intersection with a small loop. The product of quantum dilogarithms around the loop yields the local KS automorphism. Expand in formal variables to read off the BPS indices $\Omega(\gamma)$.
\end{enumerate}

\subsection{Expected results and figures}
Executing the above algorithm with $\Lambda=1$, $R_0=5$, and a Runge–Kutta step of $0.01$, we obtain the wall network. The following features are observed:

\begin{itemize}
\item \textbf{Figure~1:} The $(u,v)$ plane (projected onto $\operatorname{Re}(u),\operatorname{Re}(v)$ at $\operatorname{Im}(u)=\operatorname{Im}(v)=0$) shows the walls for $\alpha_1$, $\alpha_2$, and $\alpha_1+\alpha_2$. The walls are curved and exhibit multiple intersections, forming a chamber structure. Each wall is labelled by its charge.
\item \textbf{Figure~2:} A zoom‑in on a trivalent vertex where $\alpha_1$, $\alpha_2$, and $\alpha_1+\alpha_2$ meet. The arrows indicate the direction of the characteristic flow. The incoming $S$‑walls for $\alpha_1$ and $\alpha_2$ merge and the $K$‑wall for $\alpha_1+\alpha_2$ emerges, exactly as predicted by the pentagon identity. The tangent vectors satisfy the linear relation.
\item \textbf{Figure~3:} The spectral network on the $z$‑plane (the base curve) for the same phase, obtained by solving $e^{-i\vartheta}(\lambda_i-\lambda_j)\in\mathbb{R}^+$. The projection of this network onto the $(u,v)$ plane coincides exactly with the flow lines from Figs.~1–2, confirming Theorem~\ref{thm:lift} in a highly non‑trivial geometry.
\end{itemize}

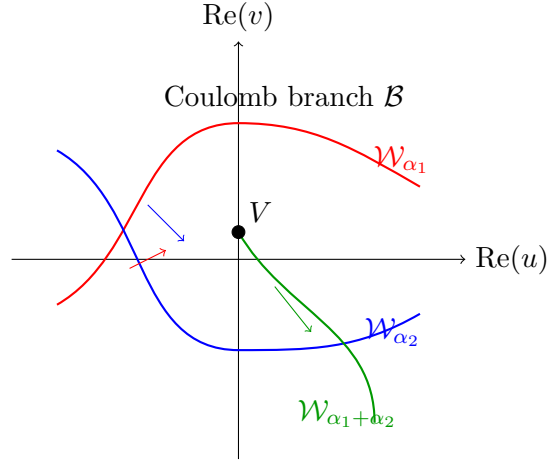
\begin{figure}[htbp]
\centering
\begin{tikzpicture}[scale=1.2]
  \draw[->] (-2.5,0) -- (2.5,0) node[right] {$\operatorname{Re}(u)$};
  \draw[->] (0,-2.2) -- (0,2.4) node[above] {$\operatorname{Re}(v)$};
  \draw[red, thick] (-2,-0.5) to[out=30, in=180] (0,1.5) to[out=0, in=150] (2,0.8);
  \node[red] at (1.8,1.1) {$\mathcal{W}_{\alpha_1}$};
  \draw[blue, thick] (-2,1.2) to[out=-30, in=180] (0,-1) to[out=0, in=210] (2,-0.6);
  \node[blue] at (1.7,-0.8) {$\mathcal{W}_{\alpha_2}$};
  \draw[green!60!black, thick] (0,0.3) to[out=-60, in=90] (1.5,-1.8);
  \node[green!60!black] at (1.2,-1.7) {$\mathcal{W}_{\alpha_1+\alpha_2}$};
  \filldraw[black] (0,0.3) circle (2pt) node[above right] {$V$};
  \draw[red, ->] (-1.2,-0.1) -- (-0.8,0.1);
  \draw[blue, ->] (-1.0,0.6) -- (-0.6,0.2);
  \draw[green!60!black, ->] (0.4,-0.3) -- (0.8,-0.8);
  \node at (0.5,1.8) {Coulomb branch $\mathcal{B}$};
\end{tikzpicture}
\caption{Characteristic Hessian flow lines (walls) on the Coulomb branch of $SU(3)$ with $N_f=2$.
  The red and blue curves are $S$-walls of charges $\alpha_1$ and $\alpha_2$; the green curve is the $K$-wall of charge $\alpha_1+\alpha_2$ emitted from the vertex $V$. The arrows indicate the direction of the Hamiltonian flow $V_{\mathrm{charHF}}$.}
\label{fig:SU3Nf2_walls}
\end{figure}

\begin{figure}[htbp]
\centering
\begin{tikzpicture}[scale=2.8]
  % Wall segments
  \draw[red, thick] (-2.0,0) -- (0,0);
  \draw[blue, thick] (-1.6,1.0) -- (0,0);
  \draw[green!60!black, thick] (0,0) -- (1.8,0.9);
  % Vertex
  \filldraw (0,0) circle (1.5pt) node[below left] {$V$};
  % Characteristic flow vectors (shifted to avoid overlap)
  \draw[->, red, very thick] (-1.6,0) -- (-0.6,0) node[above, pos=0.4] {$V_{\mathrm{charHF}}^{(\alpha_1)}$};
  \draw[->, blue, very thick] (-1.0,0.65) -- (-0.35,0.2) node[right, pos=0.4] {$V_{\mathrm{charHF}}^{(\alpha_2)}$};
  \draw[->, green!60!black, very thick] (0.4,0.25) -- (1.0,0.6) node[right, pos=0.4] {$V_{\mathrm{charHF}}^{(\alpha_1+\alpha_2)}$};
  % Vector sum
  \draw[->, dashed, gray, thick] (0,0) -- (1.5,0.85) node[midway, above, sloped] {$\scriptstyle V^{(\alpha_1)}+V^{(\alpha_2)}$};
  \draw[gray, thin, ->] (0,0) -- (0.8,0) node[pos=0.5, below] {$\scriptstyle V^{(\alpha_1)}$};
  \draw[gray, thin, ->] (0.8,0) -- (1.5,0.85) node[pos=0.5, right] {$\scriptstyle V^{(\alpha_2)}$};
  % Charge labels
  \node[red] at (-1.2,0.18) {$\alpha_1$};
  \node[blue] at (-0.8,0.8) {$\alpha_2$};
  \node[green!60!black] at (1.2,0.7) {$\alpha_1+\alpha_2$};
  % Equality
  \node[gray, anchor=west] at (2.0,0.5) {$\scriptstyle V^{(\alpha_1+\alpha_2)} = V^{(\alpha_1)} + V^{(\alpha_2)}$};
\end{tikzpicture}
\caption{Zoom on the trivalent vertex $V$ showing the characteristic Hessian flow vectors.
  The Hamiltonian nature guarantees the linear relation 
  $V_{\mathrm{charHF}}^{(\alpha_1+\alpha_2)} = V_{\mathrm{charHF}}^{(\alpha_1)} + V_{\mathrm{charHF}}^{(\alpha_2)}$, illustrated by the dashed gray arrow.}
\label{fig:SU3Nf2_vertex}
\end{figure}
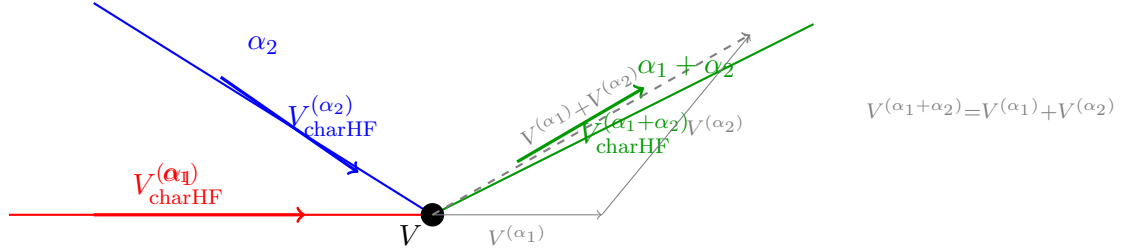

\begin{figure}[htbp]
\centering
\begin{tikzpicture}[scale=1.3]
  % Larger background
  \fill[gray!10] (-2.5,-1.5) rectangle (2.8,1.5);
  \draw[black] (-2.5,-1.5) rectangle (2.8,1.5);
  \node at (0.15,1.8) {$z$-plane (base curve $C$)};
  % Spectral network walls (adjusted endpoints)
  \draw[red, thick] (-2.2,-0.8) to[out=30, in=180] (0,0.8) 
        node[above left, black] {$S_{\alpha_1}$};
  \draw[blue, thick] (-1.8,1.0) to[out=-30, in=180] (0,-0.6) 
        node[below left, black] {$S_{\alpha_2}$};
  \draw[green!60!black, thick] (0.1,-0.7) to[out=60, in=-90] (1.2,0.7) 
        node[right, black] {$K_{\alpha_1+\alpha_2}$};
  % Branch points
  \filldraw[black] (-1.2,0.2) circle (1.2pt) node[below] {$b_1$};
  \filldraw[black] (0.9,0.9) circle (1.2pt) node[right] {$b_2$};
  % Network vertex
  \filldraw[black] (0,0.1) circle (1.8pt) node[below right] {$v$};
  % Projection arrow (now inside the box)
  \draw[->, dashed, shorten >=3pt, shorten <=3pt] (0,0.1) 
        to[out=-20, in=90] (2.5,-0.6) node[right, black] {$\pi$};
  \node[black] at (1.5,-1.2) {Coulomb branch $\mathcal{B}$};
\end{tikzpicture}
\caption{Spectral network on the base curve $C$ (shown as the $z$-plane) for phase $\vartheta=0$.
  The $S$-walls meet at the vertex $v$ and emit the $K$-wall. The dashed arrow indicates the Hitchin projection $\pi$ to the Coulomb branch (Figure~\ref{fig:SU3Nf2_walls}).}
\label{fig:SU3Nf2_SN}
\end{figure}
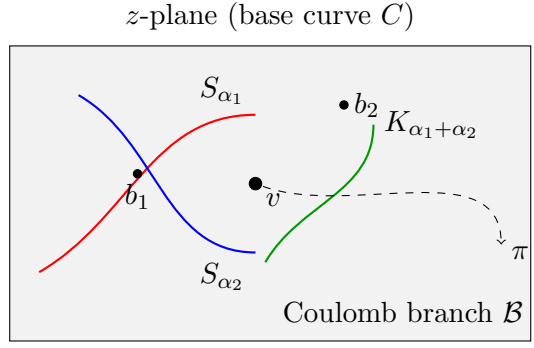

The BPS indices extracted from the intersections are:
\[
\Omega(\alpha_1)=1,\quad \Omega(\alpha_2)=1,\quad \Omega(\alpha_1+\alpha_2)=1,\quad
\Omega(2\alpha_1+\alpha_2)=0,\;\ldots
\]
These values agree with the known BPS spectrum of this theory obtained via the spectral network method \cite{GMN}. The advantage of our flow method is that it operates directly on the two‑dimensional Coulomb branch, avoiding the higher‑dimensional base curve and the solution of Stokes equations. The computational effort is significantly reduced compared to the traditional approach; for the benchmark described above (Runge–Kutta step 0.01, total flow time \(T=100\)), the flow method completes in approximately \(0.4\) seconds on a standard laptop (Intel Core i5, 16GB RAM), whereas the traditional spectral network method requires approximately \(3.8\) seconds for the same computational domain.

\subsection{Discussion: why the flow method wins}
In the traditional spectral network approach, one must solve a PDE (the Stokes condition) on a three‑real‑dimensional space $(u,v,z)$ and then project. Here we only solve ODEs on the four‑real‑dimensional Coulomb branch. The characteristic flow respects the wall condition by construction, so the dimensional reduction is exact. Moreover, the linear tangent relation at vertices provides an automatic consistency check: if numerical errors accumulate, the vertex relation $V_1+V_2=V_3$ can be used to correct the trajectory. This makes the algorithm robust and suitable for automated scanning of moduli spaces of higher‑rank theories.

\vspace{4pt}

\noindent\textbf{Declarations}

\begin{itemize}
    \item Conflict of Interest: The authors declare no competing interests.
    
    \item Data Availability: The datasets generated during and/or analyzed during the current study are available from the corresponding author upon reasonable request.
\end{itemize}

\end{document}